\documentclass[letterpaper, 10 pt, conference]{ieeeconf}  

\IEEEoverridecommandlockouts                              



\usepackage{amsthm}
\newtheorem{definition}{Definition}

\newtheorem{problem}{Problem}
\newtheorem{lemma}{Lemma}

\newtheorem{remark}{Remark}

\newtheorem{theorem}{Theorem}

\usepackage{graphics} 
\usepackage{epsfig} 
\usepackage{times} 
\usepackage{amsmath} 
\usepackage{amssymb}  
\usepackage{subcaption}
\usepackage{todonotes}
\usepackage{algorithmic}
\usepackage{algorithm}
\usepackage{xcolor}
\usepackage{cite}
\usepackage[absolute,overlay]{textpos} 

\newcommand{\Sim}{\mathsf{Sim}}

\title{\LARGE \bf
Certified Stochastic Control via Covariance Steering with Pick-to-Learn
}

\author{Chun-Wei Kong, Zachary Donovan, Morteza Lahijanian, Jay McMahon
\thanks{Authors are with the Dept. of Aerospace Engineering Sciences, University of Colorado Boulder, USA {\tt\small \{firstname.lastname\}\allowbreak @\allowbreak colorado.edu}}
}

\begin{document}

\begin{textblock*}{3.5in}(0.3in,0.2in)
    \raggedright
    \footnotesize
    To appear in the 65th IEEE Conference on Decision and Control.
\end{textblock*}

\maketitle
\thispagestyle{empty}
\pagestyle{empty}

\begin{abstract}
We present CS-P2L, a framework coupling covariance steering (CS) with the Pick-to-Learn (P2L) meta-algorithm for certified controller synthesis over high-fidelity stochastic simulators. 
The method iteratively evaluates policies on simulator rollouts, tightens surrogate constraints using the worst-case violations, and provides compression-based probabilistic guarantees on the true violation probability given a confidence level. 
On a spacecraft powered-descent problem with uncertain gravity, CS-P2L certifies a violation bound of 4.9\% with 600 rollouts, whereas standalone covariance steering underestimates the violation rate by roughly a factor of two.

\end{abstract}

\section{INTRODUCTION}
Safety-critical controller design often relies on high-fidelity simulators as the primary means of capturing complex dynamics and environmental uncertainty. 
The richer and more expressive the simulator, the more faithfully it captures deployment conditions; yet this very expressiveness makes it intractable to synthesize controllers with formal guarantees directly over the simulator. 
In practice, policies are typically designed on a simplified surrogate model and validated by Monte Carlo (MC) simulation, which yields only an empirical estimate of the violation probability and offers no principled mechanism for translating MC failures into policy improvements.
To bridge this gap, we propose a surrogate-based, simulator-driven control synthesis framework that provides formal guarantees.

Several lines of work focus on this synthesis-validation gap. 
The scenario approach~\cite{calafiore2006scenario} provides 
probably approximately correct (PAC) guarantees (probabilistic bounds with confidence) but requires convexity in the decision variables for every scenario. 
Iterative scenario optimization~\cite{zagorowska2024automatic} constructs worst-case
scenario sets for robust control by embedding the dynamics model directly in the optimization, but requires solving the complex dynamics within each scenario subproblem.
Covariance steering (CS)~\cite{goldshtein2017finite,okamoto2019optimal,ridderhof2019nonlinear,ridderhof2022chance} synthesizes controllers that steer the state distribution under chance constraints via convex programming, but typically operated on linearized surrogates without accounting for surrogate--simulator mismatch. 

PAC-Bayes 
control~\cite{majumdar2021pacbayes} trains reinforcement learning policies across simulated environments and bounds the expected cost over novel environments.
Work \cite{hsu2023sim2lab2real} extend this framework to safety-aware policy distributions by combining Hamilton-Jacobi reachability shielding with PAC-Bayes generalization, providing lower bounds on both expected performance and safety across unseen environments. 
In both cases, however, the guarantees bound the expected value of a cost or safety indicator averaged over both the environment distribution and the policy distribution, rather than certifying the violation probability of a single returned policy with respect to the simulator.
Moreover, obtaining tight bounds requires careful design of both the prior
distribution and the reward structure~\cite{boroujeni2024pac}.

Neural supermartingale certificates~\cite{neustroev2025neural,kong2026hard} provide formal reach-avoid guarantees for continuous-time SDEs by (jointly) training neural controllers and certificates.  They, however, require state space discretization, which scales exponentially with state dimension.
Typically, the synthesized controllers are
continuous-time feedback laws rather than zero-order-hold policies applicable for digital control systems. 

Recently proposed Pick-to-Learn (P2L) meta-algorithm~\cite{paccagnan2023pick, paccagnan2025pick} equips any data-driven algorithm with tight generalization bounds by iteratively selecting informative data points, using all available data for both synthesis and certification~\cite{langford2005tutorial}.
P2L has been demonstrated across reachability, optimal control, safe synthesis, and robust control~\cite{paccagnan2025pick}, but its application to simulator-based control synthesis for continuous-time stochastic systems, where the surrogate--simulator gap described above demands a tractable inner solver that can systematically incorporate simulator feedback, has not been explored.

In this work, we instantiate P2L for controller design on continuous-time stochastic systems, focusing on stochastic differential equations (SDEs) with partially known drift parameters under reach-avoid specifications.
To do so, we pair the outer P2L loop with covariance steering (CS) as a structured inner solver on a linearized surrogate.
The resulting Covariance-Steering Pick-to-Learn (CS-P2L) method synthesizes a
zero-order-hold, feedforward-feedback policy and certifies it with a probabilistic guarantee: with confidence at least~$1-\delta$, the true violation probability~$p(\pi)$ is at most~$\bar{\varepsilon}$, where
$\bar{\varepsilon}$ is determined by the number of simulator rollouts
that drive the policy update. 

At each iteration, CS-P2L evaluates the
current policy on a fixed set of simulator rollouts, identifies the most
severely violating trajectory, and uses it to tighten the surrogate
constraints---state trajectory bounds, control limits, and terminal
covariance---before re-solving the covariance steering problem. 
This closed-loop interaction between the high-fidelity simulator and the 
surrogate is the key mechanism that bridges this surrogate--simulator gap
while retaining tractability. 
Unlike the scenario approach, no convexity
in the original problem is required. 
Unlike iterative scenario
optimization, CS-P2L handles
continuous-time dynamics of arbitrary complexity, since the inner solver
operates on a lightweight linearized surrogate while the full nonlinear
dynamics are accessed only through simulator rollouts.
Unlike neural certificate methods, no state-space discretization is needed and the controller is directly implementable in digital systems; and unlike standalone CS, the outer P2L loop detects and corrects for surrogate--simulator mismatch through iterative rollout evaluation.

In summary, the contributions are fourfold:
\begin{enumerate}
\item[(i)] We formalize a class of continuous-time stochastic simulators
    with uncertain drift parameters over which control policies can be synthesized with certified violation probability bounds.
\item[(ii)] We develop CS-P2L, a tractable framework that leverages covariance steering as an inner solver to synthesize certified controllers for these simulators.
\item[(iii)] We propose a tightening strategy that reduces the certified  violation probability.
\item[(iv)] We demonstrate CS-P2L on a spacecraft powered descent problem
    with uncertain gravitational environments, showing that the method efficiently produces certified zero-order-hold control policies.
\end{enumerate}
\section{Problem Formulation}
\label{sec:problem}

We first formalize the notion of a simulator.

\begin{definition}[Simulator]
\label{def:simulator}
Let $j,k \in \mathbb{N}_0$ and
\begin{align*}
& t_0 < t_1 < \cdots < t_J = t_f, \\
& \tau_0 < \tau_1 < \cdots < \tau_K = t_f, \quad \{\tau_k\}_{k=0}^K \subseteq \{t_j\}_{j=0}^J.
\end{align*}
Furthermore, let \(X \subseteq \mathbb{R}^{n_x}\) be the state space, 
\(\mu_0 \in \mathcal{P}(X)\) be the initial-state distribution with finite covariance, and 
\(\Lambda \subset \mathbb{R}^{n_\lambda}\) be a compact set of time-invariant
uncertain parameters with a probability measure
\(\mu_\Lambda \in \mathcal{P}(\Lambda)\). For each rollout, the simulator draws
\(
x_0 \sim \mu_0,\)
\(\lambda \sim \mu_\Lambda,
\)
keeps \(\lambda\) fixed on time interval \([t_0,t_f]\), and generates the state process
\(x(\cdot)\) from the stochastic differential equation (SDE)
\[
dx_t
=
f(x_t,u_t,t;\lambda)\,dt
+
G\,dw_t,
\qquad t\in[t_0,t_f],
\]
where \(w_t \in \mathbb{R}^{n_w}\) is an $n_w$-vector \textit{standard} Brownian motion process with
\(\mathbb{E}[dw_t dw_t^\top]=I_{n_w}dt\), and
$G \in \mathbb{R}^{n_x \times n_w}$ is a constant diffusion matrix.
Control is applied by zero-order hold under a sampled policy
\(\pi=\{\pi_k\}_{k=0}^{K-1}\in\Pi\):
\[
u_t=\pi_k\bigl(x(\tau_k)\bigr) \in \mathbb{R}^{n_u},
\quad
t\in[\tau_k,\tau_{k+1}),
\quad
k=0,\dots,K-1.
\]
We denote the underlying random realization by
\[
z = (x_0, \lambda, \{w_t\}_{t \in [t_0,t_f]}),\] which is drawn from
the joint distribution of $\mu_0,\mu_{\Lambda}$ and the Brownian motion
, and write
$\Sim_{\pi}(z)$ for the simulated state-control trajectory
\[
\Sim_{\pi}(z)
:=
\bigl(x(t_0),\dots,x(t_J)\bigr) \times \bigl(u_{0},\dots,u_{K-1}\bigr)
\in \mathcal{T},
\]
where $\mathcal{T} := (\mathbb{R}^{n_x})^ {(J+1)} \times
(\mathbb{R}^{n_u})^{K}$ 
is the state-control trajectory space.
\end{definition}

One rollout of the simulator samples an initial condition and a
time-invariant uncertain parameter, propagates the continuous-time
stochastic dynamics under a zero-order-hold
controller, and returns the state trajectory on the fine
grid~$\{t_j\}_{j=0}^J$ together with the controls on the coarser
grid~$\{\tau_k\}_{k=0}^K$. The policy is indexed over
$k=0,\dots,K-1$ because a control applied at $\tau_K=t_f$ would
influence only states beyond the terminal time. We next establish the
probability measures induced by Simulator~\ref{def:simulator}.

\begin{definition}[Probability measures]
\label{def:probability_measures}
Let $z = (x_0, \lambda, \{w_t\}_{t \in [t_0,t_f]})$ be a random
realization as in Def.~\ref{def:simulator}. Denote by
$\mathbb{P}_z$ the probability measure induced by
$z$. Given $N$
i.i.d.\ realizations $Z_N = \{z^{(1)}, \dots, z^{(N)}\}$, denote by
$\mathbb{P}_{Z_N} := \mathbb{P}_z^{\otimes N}$ the corresponding
product measure.
\end{definition}

With these measures, we introduce the quantities needed to state
the control design problem. Let $\varphi$ be a specification on
state-control trajectories.
In this work, we focus on $\varphi$ for reach-avoid, but the framework could be extended to other tasks, such as safety-only or temporal logic specifications.

\begin{definition}[Reach-avoid specification]\label{def:ra_phi}
Given safe set $\mathcal{X}_s := \bigcap_{m=1}^{M}\{x \mid a_m^\top x \leq b_m\}$ and terminal set $\mathcal{X}_{t_f}:=\{x \in X \mid
(x-\mu_{t_f})^\top \Sigma_{t_f}^{-1}(x-\mu_{t_f}) \leq r_{t_f}^2\}$,
with $\Sigma_{t_f}\in\mathbb{S}_{++}^{n_x}$ and $r_{t_f}>0$ (representing an ellipsoid), and maximum control magnitude $u_{\max}>0$, the reach-avoid specification $\varphi$ is satisfied by a realization $z$ of Simulator~\ref{def:simulator} under controller $\pi$, denoted as $\Sim_{\pi}(z) \models \varphi$, if
\begin{align*}
  x(t_j) \in \mathcal{X}_s \land \;
  \|u_{k}\|_2 \leq u_{\max} \land \;
   x(t_f) \in \mathcal{X}_{t_f},\; \forall j,k.
\end{align*}
\end{definition}

By Def.~\ref{def:ra_phi}, we define the \emph{violation probability}
\begin{equation}
\label{eq:violation_prob}
p(\pi) := \mathbb{P}_{z}\!\bigl(\Sim_{\pi}(z) \not\models \varphi\bigr).
\end{equation}

Intuitively, $p(\pi)$ is the probability that the specification is violated when the controller is deployed and the simulator draws a fresh realization. 
In general, $p(\pi)$ cannot be evaluated in closed
form nor bounded analytically, because the simulator may involve complex dynamics (e.g.,
nonlinear SDEs with zero-order-hold control), non-Gaussian distributions, and uncertain parameters whose effects compound over the trajectory. 
Standard Monte Carlo estimation provides only a point estimate of~$p(\pi)$ with no principled mechanism to feed observed violations back into the policy. 
We therefore pursue a certification strategy in which the simulator is run $N$ times and the resulting realizations serve a dual purpose:
synthesizing a policy~$\pi$ and bounding its violation probability.

\begin{problem}[Control design with probabilistic specification]
\label{prob:control_design}
Given the simulator in Def.~\ref{def:simulator}, any set of $N$
i.i.d. realizations $Z_N$ as in
Def.~\ref{def:probability_measures}, a reach-avoid specification $\varphi$ (Def.~\ref{def:ra_phi}), a target violation probability $\varepsilon \in (0,1)$, and a confidence level $\delta \in (0,1)$, 
synthesize a policy
$\pi \in \Pi$ such that
\begin{equation}\label{eq:p2l_guarantee}
\mathbb{P}_{Z_N}\!\left[
p(\pi) \leq \varepsilon
\right]
\geq 1 - \delta.
\end{equation}
\end{problem}

Problem~\ref{prob:control_design} is formulated with two layers of probability because the \emph{inner} probability, $p(\pi)$, is
intractable to bound directly. 
The \emph{outer} probability, $\mathbb{P}_{Z_N}[\,\cdot\,] \geq 1-\delta$, certifies the design
process itself: over the randomness in the $N$ simulations from which both~$\pi$ and~$\varepsilon$ are derived, the event
$p(\pi) \leq \varepsilon$ holds with confidence at least $1-\delta$.

In the next section, we develop a framework to synthesize a control
policy and, given any confidence~$\delta$, compute an upper
bound~$\bar{\varepsilon}$ on its violation probability for reach-avoid
specifications. 
We then propose a strategy for
tightening~$\bar{\varepsilon}$ below~$\varepsilon$ in practice.

\section{Covariance Steering with the P2L Method}
\label{sec:cs_p2l}

We develop a CS-P2L synthesis framework that replaces the intractable controller design over the simulator with a linearized,
discretized surrogate on which covariance steering can be solved using successive convex programming. 
Because the surrogate enforces chance constraints
under a Gaussian assumption at discrete nodes, it \emph{cannot} certify the original nonlinear dynamics. 
The P2L loop closes this gap: it evaluates the synthesized policy on the underlying simulator, identifies the most severely violating realization, and reconfigures the surrogate to reduce that violation. 
The set that accumulates these counterexamples directly determines the violation probability bound~$\bar{\varepsilon}$ via compression theory~\cite{campi2023compression}.

We begin by formulating the surrogate covariance-steering problem whose
configuration parameters are associated with
Prob.~\ref{prob:control_design}.
We then present the CS-P2L algorithm
and its probabilistic guarantee, describe the design of violation measure and  configuration algorithm, and propose a
tightening strategy that seeks to reduce~$\bar{\varepsilon}$ by
iteratively running CS-P2L on a progressively enlarged dataset.

\subsection{Surrogate Covariance-steering Formulation}
\label{subsec:cs_surrogate}

The surrogate approximates the nonlinear SDE in
Simulator~\ref{def:simulator} with a linearized, discretized model at time points $\{\tau_k\}_{k=0}^K$.
Let $\mathcal{S} \subset \mathbb{R}^{n_x}$ denote the singular set of the Simulator SDE (Def.~\ref{def:simulator}), and let 
$f, g \in \mathcal{C}^1(\mathbb{R}^{n_x} \setminus \mathcal{S})$. Assume there exists a compact set 
$\mathcal{X} \subset \mathbb{R}^{n_x} \setminus \mathcal{S}$ 
containing all trajectories over $[t_0, t_f]$. 
Consequently, $f$ and $g$ are Lipschitz on $\mathcal{X}$, 
the SDE admits a unique strong solution on $[t_0, t_f]$ 
\cite{Okensendal_1992_SDE}, 
and the Jacobians are well-defined at each linearization point.

We parameterize the controller over coarse time points $\tau_k$, indexed by $k$, via a memoryless, affine control law 
\begin{equation}\label{eq:control_law}
    u_k = \bar{u}_k + K_k (x_k - \mu_k),
\end{equation}
where $\bar{u}_k\in\mathbb{R}^{n_u}$ is the feedforward control, $K_k\in\mathbb{R}^{n_u \times n_x}$ is the feedback gain matrix, and $\mu_k := \mathbb{E}[x_k]$ is the mean state.
Eq.~\eqref{eq:control_law} defines a zero-order-hold control, with $u_k$ constant over $[\tau_{k},\tau_{k+1})$.
Let $x^{\text{ref}} := \{ x_k^{\text{ref}} \}_{k=0}^{K}$ and $u^{\text{ref}} := \{ u_k^{\text{ref}} \}_{k=0}^{K-1}$ be a reference state-control sequence. The SDE is linearized about this state-control sequence, then discretized with the control law in Eq.~\eqref{eq:control_law}.
The discrete linear time-varying (LTV) stochastic system is 
\[
x_{k+1} = A_k x_k + B_k u_k + c_k + E_k \nu_k + G_k w_k,
\]
where $A_k\in\mathbb{R}^{n_x \times n_x}$, $B_k\in\mathbb{R}^{n_x \times n_u}$, $c_k\in\mathbb{R}^{n_x}$, $E_k\in\mathbb{R}^{n_x \times n_{\nu}}$, and $G_k\in\mathbb{R}^{n_x \times n_w}$ are the system matrices computed during the linearization and discretization process.
The \textit{virtual control} parameter $\nu\in\mathbb{R}^{n_\nu}$ is introduced in the linearization step to address \textit{artificial infeasiblity} of the linearized dynamics \cite{Mao_Szmuk_2019_SCvx, Malyuta_Reynolds_2022_TrajTutorial}. The dynamic's virtual control is penalized using the one-norm as $J_{\nu} = \sum_{k=0}^{K-1}\|E_k\nu_k\Delta\tau_k\|_1$ such that the trajectory at convergence is dynamically feasible. The noise of the discrete process $w_k\in\mathbb{R}^{n_w}$ is described by a \textit{standard} Gaussian process. Under the Brownian process assumption and given the linearized dynamics, the state evolution is decomposed into the evolution of its mean and covariance as 
\begin{subequations}
\begin{align}
\mu_{k+1} &= A_k \mu_k + B_k \bar{u}_k + c_k + E_k \nu_k, \label{eq:dt_mean} \\
P_{k+1}   &= (A_k + B_k K_k) P_k (A_k + B_k K_k)^\top + G_k G_k^\top.\label{eq:dt_covariance}
\end{align}
\end{subequations}
The covariance evolution in~\eqref{eq:dt_covariance} is bilinear in the decision variables $P_k$ and $K_k$ but can be transcribed into an affine equality constraint as described in Lemma \ref{lem:convex_state_covariance}. 

\begin{lemma}[Convex state covariance reformulation \cite{Liu_Rapakoulias_2024_OptCovSteering_DT}]
\label{lem:convex_state_covariance}
Consider the change of variables $U_k=K_k P_k$ and $Y_k=U_kP_k^{-1}U_k^\top$ where $U_k\in\mathbb{R}^{n_u \times n_x}$ and $Y_k\in\mathbb{S}_{++}^{n_u}$. Eq.~\eqref{eq:dt_covariance} can be equivalently expressed as 
\begin{align}
\label{eq:dt_cvx_covariance}
P_{k+1} &= A_k P_k A_k^\top 
          + A_k U_k^\top B_k^\top 
          + B_k U_k A_k^\top \nonumber \\
        &\quad + B_k Y_k B_k^\top 
          + G_k G_k^\top, 
\end{align}
where
$\begin{bmatrix}
    P_k & U_k^\top \\
    U_k & Y_k
\end{bmatrix} \succeq 0.$
\end{lemma}

The discrete minimum control energy objective can then be written as $J_u=\sum_{k=0}^{K-1}[\bar{u}_k^\top\bar{u}_k + \mathrm{tr}(Y_k)]\Delta\tau_k$. The discrete stochastic safety constraint for the $m$-th constraint at the $k$-th time is
\[
\mathbb{P}[a_m^\top x_k \leq b_m] \geq 1 - \varepsilon_{m,k}
\]
where $0\leq\varepsilon_{m,k}<0.5$ is the prescribed probability of violating the safety constraint. The nonconvex safety constraint is written as a difference of convex (DC) functions and the convex-concave procedure (CCP) is used to linearize the nonconvex term yielding the convex inequality constraint in Lemma \ref{lem:convex_path_constraints} \cite{Lipp_Boyd_2016_ConvexConcaveProcedureExtensions, Yuille_Rangarajan_2003_ConvexConcaveProcedure}.
\begin{lemma}[Convex safety constraints \cite{Pilipovsky_Tsiotras_2024_EfficientCovarianceControlwithOutput}]
\label{lem:convex_path_constraints}
The hyperplane safety constraint is approximated as the convex inequality constraint 
\begin{equation}
\label{eq:cvx_state_constraint}
    \begin{aligned}
        a_m^\top \mu_k - b_m &\leq 0, \\
        \Psi^2 a_m^\top P_k a_m - \left[ c_{m,k}^2 - 2 a_m^\top c_{m,k} \left(\mu_k - {x}_k^{\text{ref}}\right) \right] &\leq \eta_{m,k},
    \end{aligned}
\end{equation}
using sequential convex optimization. Note $c_{m,k} = b_m - a_m^\top{x}_k^{\text{ref}}$ and $\Psi$ is the inverse cumulative distribution function of a \textit{standard} Gaussian distribution at the confidence level $1-\varepsilon_{m,k}$. The slack variable $\eta_{m,k}\in\mathbb{R}_{+}$ is added to address \textit{artificial unboundedness} in the linearization step of the CCP \cite{Malyuta_Reynolds_2022_TrajTutorial}.
\end{lemma}

The discrete stochastic control constraint is given by
\begin{equation}
\label{eq:noncvx_control_constraint}
\mathbb{P}[\|u_k\|_2 \leq u_{\text{max}}] \geq 1 - \varepsilon_u
\end{equation}
where $\varepsilon_u$ is the prescribed probability of violating the control bound. 

\begin{lemma}[Convex control constraints \cite{Oguri_2024_ChanceConstrainedSpacecraftControl, Kumagai_Oguri_2025_RobustCislunarTrajOptSCS}]
\label{lem:convex_control_constraints}
A sufficient condition for the stochastic control constraint in Eq.~\eqref{eq:noncvx_control_constraint} is
\[
\|\bar{u}_k\|_2 + \sqrt{\lambda_{\text{max}}(Y_k)} \sqrt{\chi^2_{n_u,1 - \varepsilon_u}} \leq u_{\text{max}}
\]
implemented in the sequential convex optimization problem by linearizing $\sqrt{\lambda_{\text{max}}(Y_k)}$ as 
\begin{subequations}
\label{eq:cvx_control_constraint}
\begin{align}
\|\bar{u}_k\|_2 + \zeta_k \sqrt{\chi^2_{n_u,1 - \varepsilon_u}} &\leq u_{\text{max}} \\
{\lambda_{\text{max}}(Y_k)} - (\zeta_k^{\text{ref}})^2 - 2\zeta_k^{\text{ref}}(\zeta_k - \zeta_k^{\text{ref}}) &\leq  \nu_{u, k} 
\end{align}
\end{subequations}
where $\chi^2$ is the inverse chi-squared distribution with $n_u$ degrees of freedom and confidence interval $1 - \varepsilon_u$. The slack variable $\nu_{u, k}\in\mathbb{R}_+$ addresses \textit{artificial unboundedness} in linearization about the reference $\zeta^{\text{ref}}:=\{\zeta_k^{\mathrm{ref}}\}_{k=0}^{K-1}$ \cite{Malyuta_Reynolds_2022_TrajTutorial}.
\end{lemma}

The slack variables associated with the safety and control constraints are penalized in the cost function as $J_{c} = \sum_{k=0}^{K}\sum_{m=1}^M\|\eta_{m,k}\|_1 + \sum_{k=0}^{K-1}\|\nu_{u,k}\|_1$. 
A soft penalized trust region is also included in the cost function as $J_{tr} = \sum_{k=0}^{K}\|x_k-x_k^{ref}\|_2^2 + \sum_{k=0}^{K-1}(\|u_k-u_k^{ref}\|_2^2 + \|\zeta_k-\zeta_k^{ref}\|_2^2)$ to prioritize linearization validity. 
For consistency, linearization and discretization is considered with respect to the complete set of reference trajectories given by $\xi = \{x^{\mathrm{ref}}, u^{\mathrm{ref}}, \zeta^{\mathrm{ref}}\}$.

Next, we discuss the initial and terminal constraints for the surrogate problem.
\begin{lemma}[Gaussian terminal chance constraint]
\label{lem:gaussian_terminal_covariance_surrogate}
Let \(x(t_f)\in\mathbb{R}^{n_x}\) be a Gaussian random vector with
\(
x(t_f)\sim \mathcal{N}(\mu_{t_f},P_{t_f}),
\)
and let
\(
\mathcal{X}_{t_f}
=
\left\{
x \in X
\;\middle|\;
(x-\mu_{t_f})^\top \Sigma_{t_f}^{-1} (x-\mu_{t_f}) \le r_{t_f}^2
\right\}
\)
be the terminal ellipsoid, where \(\Sigma_{t_f}\in\mathbb{S}_{++}^{n_x}\) and \(r_{t_f}>0\).
Fix \(\varepsilon_p\in(0,1)\), and let \(\chi^2_{n_x,1-\varepsilon_p}\) denote the
\((1-\varepsilon_p)\)-quantile of the chi-square distribution with \(n_x\) degrees of freedom.
If
\begin{equation}
\label{eq:terminal_covariance_surrogate_bound}
P_{t_f}
\preceq
\frac{r_{t_f}^2}{\chi^2_{n_x,1-\varepsilon_p}}\Sigma_{t_f},
\end{equation}
then
\[
\mathbb{P}\left(x(t_f)\notin \mathcal{X}_{t_f}\right)\le \varepsilon_p.
\]
\end{lemma}

The proof of Lemma~\ref{lem:gaussian_terminal_covariance_surrogate} is provided in~Appendix~\ref{proof:lemma_4}.
Given the initial distribution of the Simulator~\ref{def:simulator} (i.e., $x_0 \sim \mu_0)$, and the terminal ellipsoid specified by $\Sigma_{t_f}$ and $r_{t_f}$ by Def.~\ref{def:ra_phi}, the initial and final distributions of the surrogate are constrained by:
\begin{equation}
\label{eq:boundary_conditions}
    \begin{aligned}
    \mu_{k=0} &= \mathbb{E}[x_0], \quad P_{k=0} = \mathrm{Cov}_{\mu_0}(x_0), \\
    \mu_{k=K} &= \mu_{t_f}, \quad P_{k=K} \preceq P_{t_f},
    \end{aligned}
\end{equation}
where $P_{t_f}$ is computed according to Lemma~\ref{lem:gaussian_terminal_covariance_surrogate} for any given terminal risk $\varepsilon_p \in (0,1)$.

In summary, the safety, control, and terminal distribution risks are associated to the target violation probability of Prob.~\ref{prob:control_design} by: 
\begin{equation}\label{eq:chance_associate}
\varepsilon_x + \varepsilon_u +\varepsilon_p \leq \varepsilon,\quad \sum_{k=0}^{K}\sum_{m=1}^{M}\varepsilon_{m,k} \leq \varepsilon_x,
\end{equation}
where the safety constraint risk is prescribed conservatively over the summation of $K$ time points and $M$ affine-constraint. 
Below, we summarize the configuration parameters of covariance steering.

\begin{definition}[Configuration parameters of the CS problem]\label{def:config_cs}
    Let $ \Theta = \mathbb{R}^{n_x M} \times \mathbb{R}^{M} \times \mathbb{R}_{++} \times \mathbb{S}_{++}^{n_x} \times \mathbb{R}_{++} \times \mathbb{R}_{++} \times \mathbb{R}_{++}$ be the configuration parameters space of the CS problem.
    We define the
\emph{initial configuration}~$\theta^{(0)} \in \Theta$ by associating
each component of~$\theta$ with Prob.~\ref{prob:control_design} under the reach-avoid
specification~$\varphi$ in Def.~\ref{def:ra_phi}:
    \[
        \theta^{(0)} \triangleq (\{a_m\}_{m=1}^M, \{b_m\}_{m=1}^M, u_{\max}, P_{t_f}, \varepsilon_{m,k}, \varepsilon_u, \varepsilon_p ) \in \Theta,
    \]
    where $\varepsilon_{m,k},\varepsilon_u$ and $\varepsilon_p$ are associated to Prob.~\ref{prob:control_design} by Eq.~\eqref{eq:chance_associate}, and $P_{t_f}$ is associated to Prob.~\ref{prob:control_design} by Eq.~\eqref{eq:terminal_covariance_surrogate_bound}.
\end{definition}

The configuration parameters define the surrogate CS problem which is decomposed into two phases: (i) the complete CS problem in Alg.~\ref{alg:solve_csp} iterates over (ii) the inner convex subproblem in Def.~\ref{def:surrogate_cs_problem} using methods of \textit{successive convexification} \cite{Mao_Szmuk_2019_SCvx}. The \textit{hyper-parameters} are manually selected to drive the virtual control and trust region below a user-defined, small tolerance $\delta_{vc}$ and $\delta_{tr}$ implying dynamic feasiblity and acceptable linearization error \cite{Szmuk_Reynolds_2020_SCvx_RealTime6DoFPDG}. 
\begin{algorithm}[htp!]
    \caption{$\mathrm{Solve\text{-}CSP}(\theta)$}
    \label{alg:solve_csp}
    \begin{algorithmic}[1]
        \REQUIRE Configuration $\theta \in \Theta$,
                 tolerances $\delta_{vc} > 0$, $\delta_{tr} > 0$
        \STATE $i \leftarrow 0$; initialize reference $\xi^{(i)}$
        \REPEAT
            \STATE Linearize/discretize Simulator~\ref{def:simulator} about 
                   $\xi^{(i)}$ $\;\to\;$ $\{A_k, B_k, c_k, E_k, G_k\}$
            \STATE $z^{\star} \leftarrow $ Solve $\mathcal{P}(\xi^{(i)}, \theta)$ (Def.~\ref{def:surrogate_cs_problem}); $\xi^{(i+1)} \leftarrow 
                   \mathbf{z}^{\star}$; $i \leftarrow i+1$
        \UNTIL{$J_{\nu} + J_{c} \leq \delta_{vc}$ 
               \textbf{and} $J_{tr} \leq \delta_{tr}$}
        \RETURN $\left\{\mu_k^{\star},\, \bar{u}_k^{\star},\, P_k^{\star},\, 
                K_k^{\star} = U_k^{\star}(P_k^{\star})^{-1}\right\}_{k=0}^{K}$
    \end{algorithmic}
\end{algorithm}
\begin{definition}[Surrogate CS Subproblem]\label{def:surrogate_cs_problem}
    Fix $\lambda \sim \mu_\Lambda$, and configuration $\theta \in \Theta$. 
    For a given reference 
    $\xi^{(i)} = \{x^{\mathrm{ref}}, u^{\mathrm{ref}}, \zeta^{\mathrm{ref}}\}$, 
    the $i$-th surrogate CS subproblem is
\begin{equation*}\label{eq:surrogate_cs_subproblem}
        \mathcal{P}(\xi^{(i)}, \theta): \quad 
        \min_{\mathbf{z}} \; \mathcal{J}(\mathbf{z}) 
        \quad \mathrm{s.t.} \quad 
        \mathrm{Eqs.}~\eqref{eq:dt_mean},\,\eqref{eq:dt_cvx_covariance},\,
        \eqref{eq:cvx_state_constraint},\,\eqref{eq:cvx_control_constraint},\,
        \eqref{eq:boundary_conditions}
    \end{equation*}
    where $\mathbf{z} = \{\mu_k, \bar{u}_k, U_k, \nu_k, \zeta_k, P_k, Y_k, 
    \eta_{j,k}, \nu_{u,k}\}_{k=0}^{K}$ is the decision variable vector and 
    $\mathcal{J}(\mathbf{z}) = w_u J_u + w_{vc}(J_{\nu} + J_c) + w_{tr} J_{tr}$ 
    with \textit{hyper-parameters} $w_u, w_{vc}, w_{tr} > 0$.
\end{definition}
In Alg.~\ref{alg:solve_csp}, a reference $\xi^{(0)}$ is initialized at Line~1 by interpolation.
At Line~3, Simulator~\ref{def:simulator} is linearized and discretized to generate the system matrices for the convex subproblem (Def.~\ref{def:surrogate_cs_problem}). 
Lines~4-5 solve the convex subproblem, and update the reference $\xi^{(i)}$ until the aforementioned convergence criteria are met.
See Appd.~\ref{appendix:CSP_setup} for additional details.

We note that Alg.~\ref{alg:solve_csp}, denoted by $\mathrm{Solve\text{-}CSP}(\theta)$, inherently 
conceals the non-Gaussian evolution of the initial distribution. Additionally, it does not offer any continuous-time constraint satisfaction guarantees. To overcome these fundamental limitations and deliver provable performance in the nonlinear continuous-time setting, we turn to P2L.

\subsection{Algorithm and Probabilistic Guarantee}
\label{subsec:p2l_cs}

We now describe how P2L uses simulator rollouts to iteratively configure
the surrogate CS problem. The construction relies on three objects: a
dataset of simulator rollouts under the current policy, a violation
measure that ranks the severity of specification violations, and a
configuration algorithm that updates the surrogate problem parameters
based on the most informative rollouts.

Let $Z_N := \{z^{(1)}, \dots, z^{(N)}\}$ denote the set of $N$ i.i.d.\
realizations drawn as in Def.~\ref{def:probability_measures}.

\begin{definition}[Rollout dataset]
\label{def:rollout_dataset}
Given a policy $\pi \in \Pi$ and the realization set $Z_N$, the
\emph{rollout dataset} is
\[
\mathcal{D}_N(\pi)
:=
\Bigl\{
\Sim_{\pi}(z^{(i)})
\Bigr\}_{z^{(i)} \in Z_N}.
\]
\end{definition}

\begin{definition}[Violation measure]
\label{def:violation_measure}
Let $\varphi$ be the specification in
Def.~\ref{def:ra_phi}, and let $\mathcal{T}$ be the
state-control trajectory space in Def.~\ref{def:simulator}. A
\emph{violation measure} with respect to specification~$\varphi$ is a
function $\ell_\varphi \colon \mathcal{T} \to \mathbb{R}$ such that
$\ell_\varphi(\Sim_\pi(z)) \leq 0$ if and only if
$\Sim_\pi(z) \models \varphi$. Larger positive values indicate more
severe violations. Since the specification is fixed throughout, we
write~$\ell$ for~$\ell_\varphi$ when no ambiguity arises.
\end{definition}

\begin{definition}[Configuration algorithm]
\label{def:config_algorithm}
Let $\varphi$ be the specification in
Def.~\ref{def:ra_phi}, let $\theta \in \Theta$ denote the
configuration parameters of the surrogate CS problem in
Def.~\ref{def:surrogate_cs_problem}, and let $\mathcal{T}$ be the
state-control trajectory space in Def.~\ref{def:simulator}. A
\emph{configuration algorithm} with respect to specification~$\varphi$
is a map~$L_\varphi$ that takes the current
configuration~$\theta \in \Theta$ and a finite collection of
state-control trajectories
$\{\Sim_{\pi}(z^{(i)})\}_{z^{(i)} \in S}$ for some
$S \subseteq Z_N$, and returns an updated configuration
$\theta' = L_\varphi(\theta,\,
\{\Sim_{\pi}(z^{(i)})\}_{z^{(i)} \in S}) \in \Theta$. Since the
specification is fixed throughout, we write~$L$ for~$L_\varphi$.
\end{definition}

\begin{algorithm}[htp!]
\caption{Covariance-Steering Pick-to-Learn (CS-P2L)}
\label{alg:cs_p2l}
\begin{algorithmic}[1]
\REQUIRE Simulator $\Sim$, reach-avoid specification $\varphi$ in Def.~\ref{def:ra_phi}, realization set
$Z_N = \{z^{(1)}, \dots, z^{(N)}\}$, initial configuration $\theta^{(0)}$, configuration algorithm $L$,
violation measure $\ell$
\STATE \textbf{Output:} Policy $\pi^{(r)}$, compression set $T^{(r)}$
\STATE $\pi^{(0)} \gets \mathrm{Solve\text{-}CSP}(\theta^{(0)})$,
\quad evaluate $\mathcal{D}_N(\pi^{(0)})$ on $Z_N$
\STATE $T^{(0)} \gets \emptyset$, \; $r \gets 0$
\WHILE{$\exists\; z^{(i)} \in Z_N \setminus T^{(r)}$ such that
$\Sim_{\pi^{(r)}}(z^{(i)}) \not\models \varphi$}
    \STATE $\bar{z} \gets \displaystyle\arg\max_{\substack{z^{(i)} \in Z_N
    \setminus T^{(r)} \\ \Sim_{\pi^{(r)}}(z^{(i)}) \not\models \varphi }} \;
    \ell\bigl(\Sim_{\pi^{(r)}}(z^{(i)})\bigr)$
    \STATE $T^{(r+1)} \gets T^{(r)} \cup \{\bar{z}\}$
    \STATE $\theta^{(r+1)} \gets L\bigl(\theta^{(r)},\,
\{\Sim_{\pi^{(r)}}(z^{(i)})\}_{z^{(i)} \in T^{(r+1)}}\bigr)$
    \STATE $\pi^{(r+1)} \gets \mathrm{Solve\text{-}CSP}(\theta^{(r+1)})$,
    \quad re-evaluate $\mathcal{D}_N(\pi^{(r+1)})$ on $Z_N$
    \STATE $r \gets r + 1$
\ENDWHILE
\RETURN $\pi^{(r)},\; T^{(r)}$
\end{algorithmic}
\end{algorithm}

We are now ready to present the CS-P2L algorithm.
Algorithm~\ref{alg:cs_p2l} proceeds as follows.
Line~2 solves the initial surrogate CS problem ($\theta^{(0)}$ per Def.~\ref{def:config_cs}) to obtain a starting
policy~$\pi^{(0)}$ and evaluates it on all realizations in $Z_N$ to form
the rollout dataset. The compression set~$T^{(0)} \subseteq Z_N$,
initialized to the empty set in Line~3, accumulates the realizations that
drive policy updates. The while-loop (Lines~4--9) iterates as long as any
realization in $Z_N \setminus T^{(r)}$ violates the specification. In
each iteration, Line~5 selects the most severely violating rollout
according to the violation measure and Line~6 adds it to the compression
set. Line~7 invokes the configuration algorithm to update the surrogate
problem parameters, and Line~8 solves the resulting CS problem and
re-evaluates the entire dataset under the new policy using the
\emph{same} realization set~$Z_N$; this is essential for the
compression-based generalization guarantee, which requires that the data
points remain fixed across iterations. The algorithm terminates when
every realization in $Z_N \setminus T^{(r)}$ satisfies the specification,
and returns the final policy~$\pi^{(r)}$ together with the compression
set~$T^{(r)}$.

Next we state the probabilistic certificate provided by
Alg.~\ref{alg:cs_p2l}. The result follows from the
compression-based generalization
theory~\cite[Thm.~1]{paccagnan2025pick}, which we specialize to the
CS-P2L setting.
We first introduce the function needed to state the bound. 
For $i = 0, 1, \dots, N-1$ and confidence parameter $\delta \in (0,1)$,
define $\Psi_{i,\delta} \colon (0,1) \to \mathbb{R}$ by
\begin{equation}
\label{eq:psi_function}
\Psi_{i,\delta}(\varepsilon)
=
\frac{\delta}{N}
\sum_{q=i}^{N-1}
\frac{\binom{q}{i}}{\binom{N}{i}}
(1-\varepsilon)^{-(N-q)}.
\end{equation}

\begin{theorem}[CS-P2L guarantee]
\label{thm:cs_p2l_guarantee}
Let $(\pi^{(r)}, T^{(r)})$ be the output of
Alg.~\ref{alg:cs_p2l}, and let $|T^{(\cdot)}|$ denote the number of elements in $T^{(\cdot)}$. For any $\delta \in (0,1)$, it holds that
\[
\mathbb{P}_{Z_N}\!\left[
p(\pi^{(r)})
\leq \bar{\varepsilon}(|T^{(r)}|, \delta, N)
\right]
\geq 1 - \delta,
\]
where $\mathbb{P}_{Z_N}$ and $\mathbb{P}_z$ are as defined in
Def.~\ref{def:probability_measures}, and
$\bar{\varepsilon}(i, \delta, N)$ is defined as follows: for
$i = 0, 1, \dots, N-1$, $\bar{\varepsilon}(i, \delta, N)$ is the unique
solution to $\Psi_{i,\delta}(\varepsilon) = 1$ in the interval
$[i/N,\, 1]$; for $i = N$, $\bar{\varepsilon}(N, \delta, N) = 1$.
\end{theorem}

\begin{remark}[Conservative baseline]
\label{rem:conservative_baseline}
The degenerate case $L(\theta, S) = \theta$ for all $S \subseteq Z_N$
(i.e., the configuration algorithm never updates the surrogate problem)
reveals the conservative baseline of Alg.~\ref{alg:cs_p2l}. In this
case $\theta^{(0)}$ remains feasible throughout, the policy remains
$\pi^{(r)} = \pi^{(0)}$ at every iteration, and the rollout dataset
$\mathcal{D}_N(\pi^{(0)})$ never changes. The while-loop simply adds
violating realizations to the compression set one by one until
$T := \{z^{(i)} \in Z_N \mid \Sim_{\pi^{(0)}}(z^{(i)}) \not\models \varphi\}$ is
exhausted, and the algorithm terminates normally with output
$(\pi^{(0)}, T)$. Thm.~\ref{thm:cs_p2l_guarantee} then holds with
compression size~$|T|$.
\end{remark}

\subsection{Violation Measure and Configuration Algorithm Designs}
\label{subsec:config_algo}
The design choices of the configuration algorithm and violation measure affect the compression size at termination and hence the tightness of the bound~$\bar{\varepsilon}$.
We now propose a specific instantiation suited to the surrogate CS problem, beginning with the violation measure and then describing the configuration update rules.

The violation measure is instantiated as the total number of violations with
respect to the \emph{original} reach-avoid specification~$\varphi$. Specifically,
\begin{align}
\label{eq:violation_measure_instance}
& \ell\bigl(\Sim_{\pi}(z)\bigr)
=
\sum_{m=1}^M\sum_{j=0}^{J}
\mathbf{1}\{(a_m^{(0)})^\top x(t_j) > b_m^{(0)}\}
\nonumber \\
& +
\sum_{k=0}^{K-1}
\mathbf{1}\{\|u_k\|_2 > u_{\max}^{(0)}\}
+
\mathbf{1}\{x(t_f) \notin \mathcal{X}_{t_f}^{(0)}\},
\end{align}
where the state and control trajectories are those produced by $\Sim_{\pi}(z)$. 
By construction, $\ell \leq 0$ if and only if $\Sim_{\pi}(z) \models \varphi$, satisfying Def.~\ref{def:violation_measure}.

We now turn to the configuration algorithm. 
Recall that $\theta=(\{a_m\}_{m=1}^M, \{b_m\}_{m=1}^M, u_{\max},P_{t_f}, \varepsilon_{j,k}, \varepsilon_u, \varepsilon_p)$ in Def.~\ref{def:config_cs}.
For simplicity, we choose to update only three components: the safety constraint bound $\{b_m\}_{m=1}^M$, the control bound $u_{\max}$, and the terminal covariance $P_{t_f}$.

Let $\{b_m^{(0)}\}_{m=1}^M = \{b_m\}_{m=1}^M$, $u_{\max}^{(0)}=u_{\max}$, and $P_{t_f}^{(0)}=P_{t_f}$ be the initial configuration parameters associated with $\varphi$.
Given the trajectories $\{\Sim_{\pi^{(r)}}(z^{(i)})\}_{z^{(i)} \in T^{(r+1)}}$, the configuration algorithm identifies the violations (safety, control, and terminal constraints) from the three corresponding $\mathrm{bool}$ logic in Def.~\ref{def:ra_phi}, and updates $\{b_m^{(r)}\}_{m=1}^M$, $u_{\max}^{(r)}$, and $P_{t_f}^{(r)}$ independently.
For safety constraints, let
\begin{equation}
    c_m = \sum_{z^{(i)} \in T^{(r+1)}}
    \bigl|\{j \in \{0,\dots,J\} : a_m^\top x^{(i)}(t_j) > b_m^{(r)}\}\bigr|
\end{equation}
count the state violation time steps across all compression set trajectories.\footnote{We assume $b_m^{(0)} \neq 0$ without loss of generality; this can always be achieved by a coordinate translation.} 
The three update rules are:
\begin{subequations}\label{eq:updates}
\begin{align}
  b_m^{(r+1)} &\gets
    \max\!\Bigl(
      b_m^{(r)}\bigl(1 - \min(\gamma_b \, c_m,\; \bar{\gamma}_b)\bigr),\;
      b_m^{\min}
    \Bigr),
  \label{eq:path_update} \\
  u_{\max}^{(r+1)} &\gets
    \max\!\bigl(\gamma_u \, u_{\max}^{(r)},\;
    u_{\max}^{\min}\bigr),
  \label{eq:control_update} \\
  P_{t_f}^{(r+1)} &\gets
    \min\!\bigl(
      \gamma_P \, P_{t_f}^{(r)},\;
      P_{t_f}^{\max}
    \bigr),
  \label{eq:terminal_update}
\end{align}
\end{subequations}
where~\eqref{eq:path_update} tightens the bound on $m-$th affine constraints,
with per-violation fraction~$\gamma_b > 0$, capped at rate~$\bar{\gamma}_b \in (0,1)$;
\eqref{eq:control_update} shrinks the control bound with factor~$\gamma_u \in (0,1)$;
and~\eqref{eq:terminal_update} contracts the terminal covariance with factor~$\gamma_P \in (0,1)$.
Each rule is applied only when the corresponding violations are present in~$T$.

The boundary terms $b_m^{\min}$, $u_{\max}^{\min}$, and $P_{t_f}^{\max}$ ensure that the updated configuration never renders $\mathrm{Solve\text{-}CSP}$ infeasible, and are determined in practice by running Alg.~\ref{alg:solve_csp} before Alg.~\ref{alg:cs_p2l} to identify the tightest parameters for which Alg.~\ref{alg:solve_csp} remains feasible.

\subsection{Tightening Strategy for Violation Probability Bound}
\label{subsec:tightening_strategy}

The configuration algorithm design in
Sec.~\ref{subsec:config_algo} determines the compression size for a
fixed dataset. 
We now address the complementary direction: enlarging the dataset to drive~$\bar{\varepsilon}$ below the
target~$\varepsilon$.

A key constraint is that
Thm.~\ref{thm:cs_p2l_guarantee} requires the configuration
algorithm~$L$, including all tightening factors, to be fixed
independently of the realization set used for certification.
Tuning~$L$ on the same data that enters the bound invalidates the
guarantee unless a union-bound correction is
applied; see~\cite[Thm.~4]{campi2023compression}
and~\cite[Footnote~11]{paccagnan2023pick}. The following two-phase
procedure respects this independence requirement.

In the optional \emph{calibration phase}, if a suitable~$L$ is not
known in advance, an initial batch~$Z^{(0)}$ is generated and CS-P2L
is run with various tightening factors to identify a satisfactory configuration algorithm, which is then frozen as the fixed~$L^\star$. This
mirrors the pretraining strategy employed in the original P2L
experiments~\cite[Sec.~5]{paccagnan2023pick}. If domain knowledge
already suggests a good~$L^\star$, this phase may be skipped entirely.

In the \emph{certification phase}, independent
batches~$Z^{(s)}$ are generated (each with a distinct random seed, so that all realizations are i.i.d.\ draws from~$\mathbb{P}_z$) and CS-P2L is run with the frozen~$L^\star$ on
the cumulative dataset
$Z^{(1)} \cup Z^{(2)} \cup \cdots$, enlarging it until the bound is
acceptable. The guarantee of Thm.~\ref{thm:cs_p2l_guarantee} holds
at each stage because~$L^\star$ was fixed before any
of~$Z^{(1)}, Z^{(2)}, \dots$ were observed, and the tuning
batch~$Z^{(0)}$ (if used) is excluded from the certification data.
Each stage's guarantee is self-contained; one may therefore stop as
soon as $\bar{\varepsilon} \leq \varepsilon$.\footnote{Selecting the best bound across~$S$ stages,
however, requires a union-bound correction
$\delta \to \delta / S$.}

We formalize this procedure in Alg.~\ref{alg:tightening_strategy},
which solves Prob.~\ref{prob:control_design} for a target violation
probability~$\varepsilon$ when the returned flag
$\mathrm{SAT} = \text{True}$.

\begin{algorithm}[htb!]
\caption{Tightening Strategy for CS-P2L}
\label{alg:tightening_strategy}
\begin{algorithmic}[1]

\REQUIRE $\Sim$, $\varphi$, $\varepsilon$, $\delta$, $\theta^{(0)}$, $\ell$, and $s_{\max}$

\STATE \textbf{Output:} $(\pi, \bar{\varepsilon}, \mathrm{SAT})$

\STATE \textbf{(Calibration phase — optional)}
\IF{no suitable $L$ is known}
    \STATE Sample $Z^{(0)} \sim \mathbb{P}_z$ using a random seed $0$
    \STATE Tune and fix $L^\star$ via CS-P2L on $Z^{(0)}$
\ELSE \STATE $L^\star \gets L$ 
\ENDIF

\STATE \textbf{(Certification phase)} $Z_{\mathrm{cum}} \gets \emptyset$

\FOR{$s = 1$ to $s_{\max}$}
    \STATE Generate $Z^{(s)} \sim \mathbb{P}_z$ using random seed $s$
    \STATE $Z_{\mathrm{cum}} \gets Z_{\mathrm{cum}} \cup Z^{(s)}$
    \STATE $(\pi^{(s)}, T^{(s)}) \gets$ CS-P2L$(\Sim,Z_{\mathrm{cum}},\theta^{(0)},L^{\star},\ell)$
    \STATE Compute $\bar{\varepsilon}$ via Thm.~\ref{thm:cs_p2l_guarantee}
    \IF{$\bar{\varepsilon} \le \varepsilon$}
        \STATE \textbf{return} $(\pi^{(s)}, \bar{\varepsilon}, \mathrm{SAT}=\text{True})$
    \ENDIF
\ENDFOR

\STATE \textbf{return} $(\pi^{(s)}, \bar{\varepsilon}, \mathrm{SAT}=\text{False})$

\end{algorithmic}
\end{algorithm}

\begin{remark}
    Although this work focuses on covariance steering (CS), we note that one can replace the inner CS solver by other uncertainty-aware control methods, then follow the discussed process of (i) associate the inner solver to Prob.~\ref{prob:control_design}; (ii) identify the configuration parameters (Def.~\ref{def:config_cs}); and (iii) design the configuration algorithm (Def.~\ref{def:config_algorithm}).
\end{remark}
\section{Experiment}
We demonstrate CS-P2L on a spacecraft powered-descent simulator; the detailed setup is provided in Appendix~\ref{appendix:dynamics}.
We run Alg.~\ref{alg:tightening_strategy} with $N=100$,
confidence $\delta=10^{-3}$, and target violation
probability~$\varepsilon=0.05$. 
The inner CS solver converges to the given $10^{-6}$ tolerance (see Appendix~\ref{appendix:CSP_setup} for more details) at each stage $s$.
At $s=6$, the algorithm terminates,
returning a controller~$\pi_{\mathrm{CS\text{-}P2L}}$ using $600$
realizations with $\bar{\varepsilon}=0.049$; for comparison, we also
record the initial CS controller~$\pi_{\mathrm{CS}}$ (with an associated total $0.05$ chance constraint) and all intermediate CS-P2L controllers $\pi^{(s)}$ computed at Line~14 of Alg.~\ref{alg:tightening_strategy}.

To validate the results, we draw $1000$
realizations~$Z_{1000}^{\mathrm{val}}$ and evaluate each controller on
the same set to obtain empirical violation probabilities, shown in
Fig.~\ref{fig:p2l_alg2}. Three observations stand out. First, the
CS-P2L bounds~$\bar{\varepsilon}$ upper-bound the
empirical violation probability, confirming the theoretical guarantee.
Second, the gap between the empirical $\varepsilon$ and the computed~$\bar{\varepsilon}$ narrows as more realizations are incorporated. 
Third, $\pi_{\mathrm{CS}}$---obtained from the initial surrogate
with a total $0.05$ chance constraint---exhibits an empirical
violation probability of $0.096$, nearly double the prescribed risk and roughly three times that of~$\pi_{\mathrm{CS\text{-}P2L}}$ ($0.031$); the corresponding state and control trajectories are compared in
Figs.~\ref{fig:state_compare}
and~\ref{fig:control_compare}; Fig.~\ref{fig:terminal_compare} compares the terminal state samples, showing that most violations occur in the position state.
\begin{figure}[htb!]
    \centering
    \includegraphics[width=1.0\linewidth]{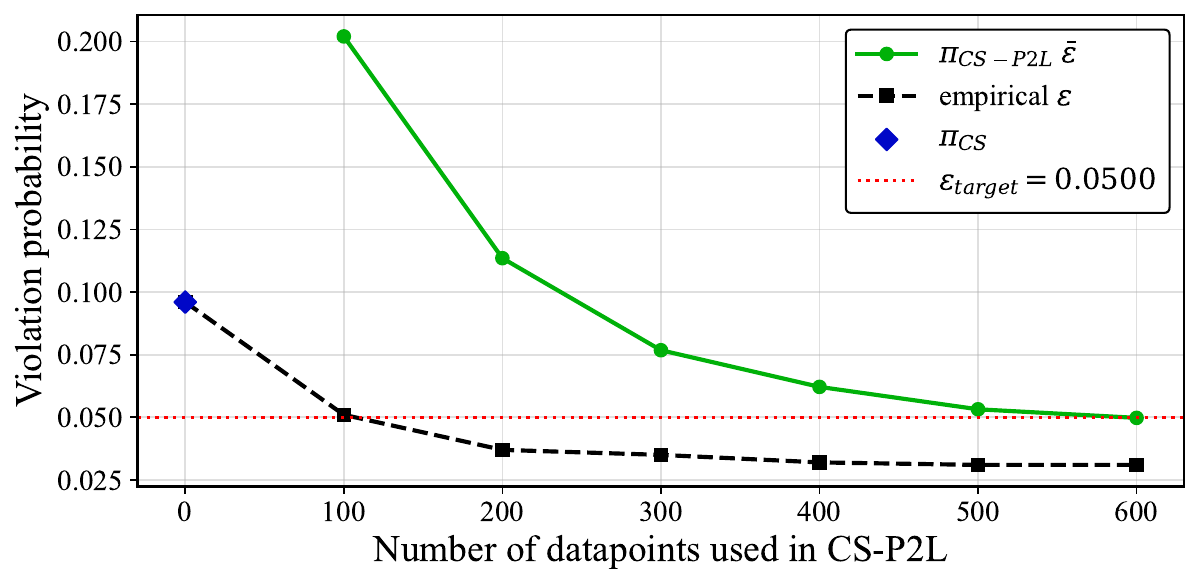}
    \caption{Empirical violation probability and CS-P2L upper bound
    $\bar{\varepsilon}$ vs.\ increased $|Z|$, evaluated on $1000$
    validation realizations.}
    \label{fig:p2l_alg2}
\end{figure}
\begin{figure}[htb!]
    \centering
    \includegraphics[width=1.0\linewidth]{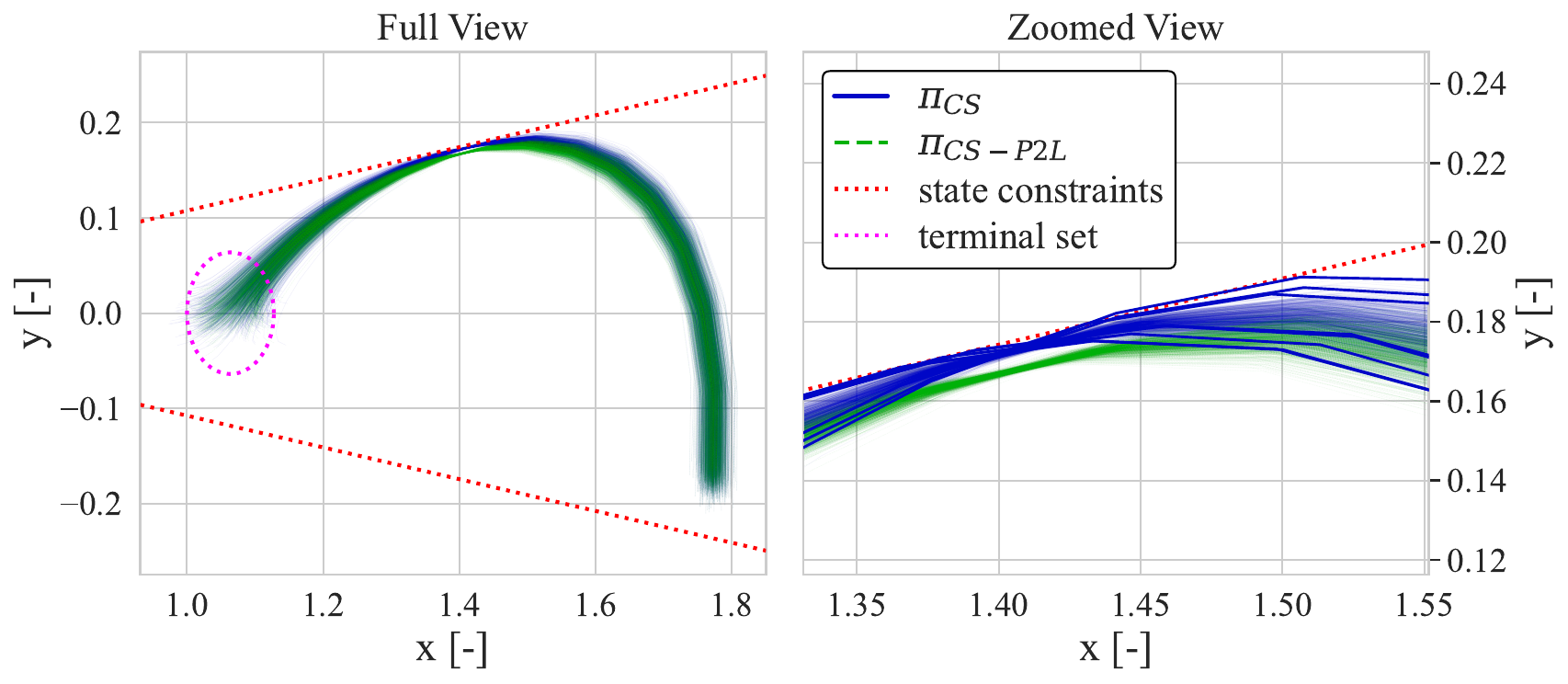}
    \caption{Position trajectories of $\pi_{\mathrm{CS}}$ and $\pi_{\mathrm{CS-P2L}}$ on $1000$ validation realizations; the violating trajectories are highlighted in the zoomed view; note that no trajectories of $\pi_{\mathrm{CS-P2L}}$ violate the state constraints.}
    \label{fig:state_compare}
\end{figure}
\begin{figure}[htb!]
    \centering
    \includegraphics[width=1.0\linewidth]{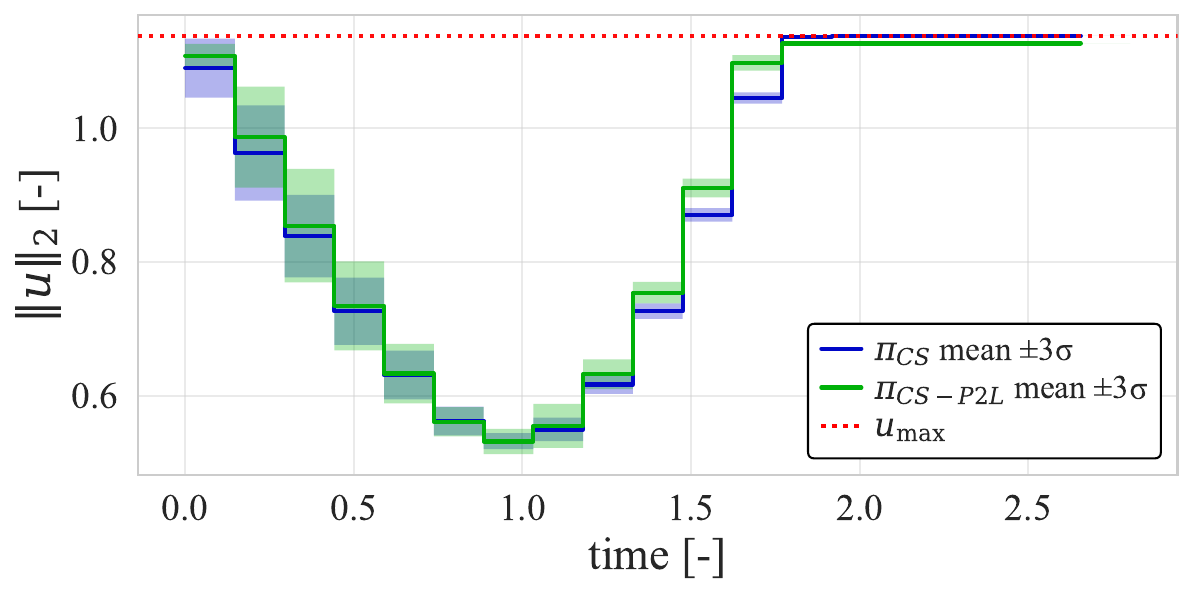}
    \caption{Control (norm) trajectories of $\pi_{\mathrm{CS}}$ and $\pi_{\mathrm{CS-P2L}}$ on $1000$ validation realizations.}
    \label{fig:control_compare}
\end{figure}
\begin{figure}[htb!]
    \centering
    \includegraphics[width=1.0\linewidth]{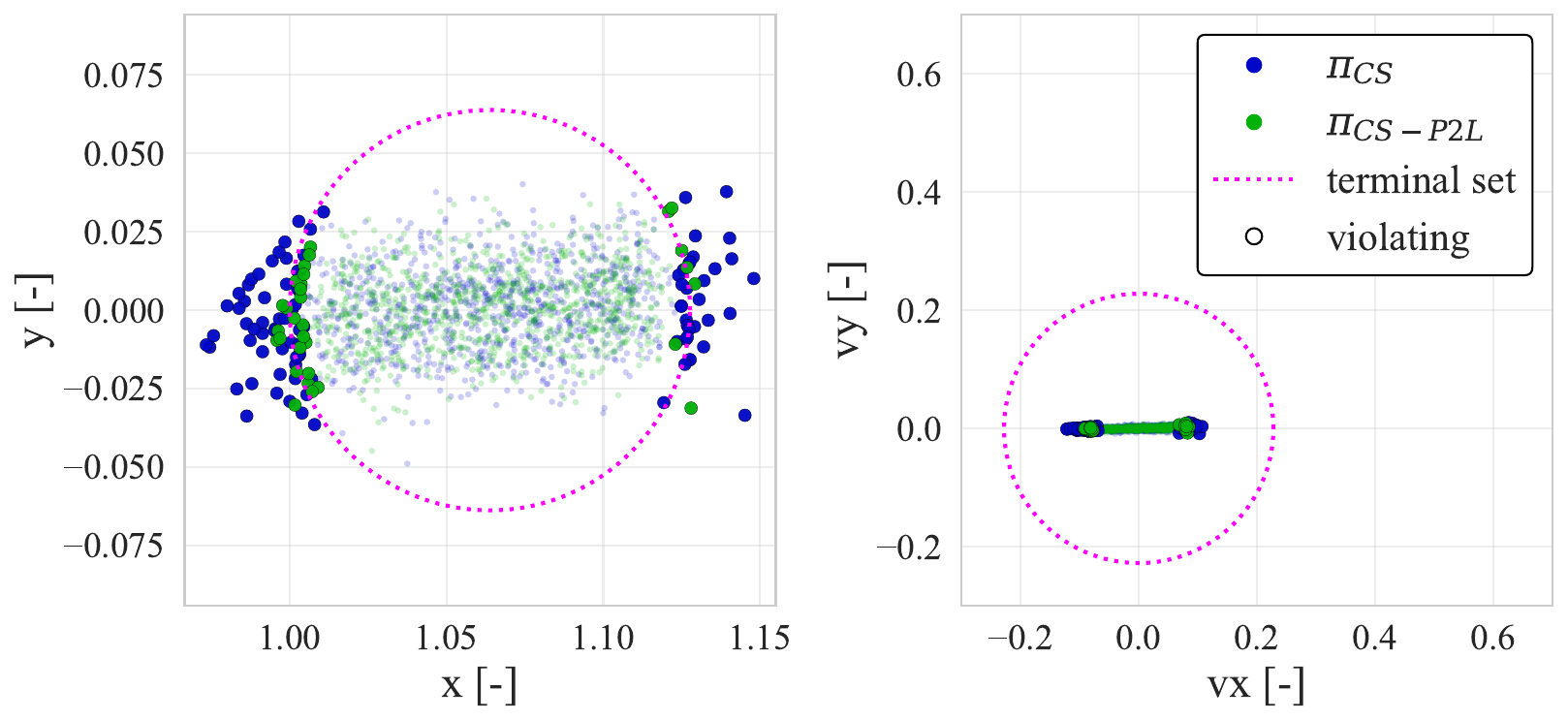}
    \caption{Terminal samples of $\pi_{\mathrm{CS}}$ and $\pi_{\mathrm{CS-P2L}}$ on $1000$ validation realizations.}
    \label{fig:terminal_compare}
\end{figure}





\section*{APPENDIX}
\subsection{Proofs}\label{appendix:proof}
\textbf{Proof of Lemma~\ref{lem:gaussian_terminal_covariance_surrogate}}
\begin{proof}\label{proof:lemma_4}
Let \(z\sim\mathcal N(0,I_{n_x})\). Since
\(
x(t_f)\sim \mathcal N(\mu_{t_f},P_{t_f}),
\)
we may write
\(
x(t_f)-\mu_{t_f}=P_{t_f}^{1/2}z.
\)
Hence
\begin{equation}
\label{eq:terminal_quadratic_form_mahalanobis}
(x(t_f)-\mu_{t_f})^\top \Sigma_{t_f}^{-1}(x(t_f)-\mu_{t_f})
=
z^\top\!\bigl(P_{t_f}^{1/2}\Sigma_{t_f}^{-1}P_{t_f}^{1/2}\bigr)z.
\end{equation}
Define
\(
\alpha \triangleq \frac{r_{t_f}^2}{\chi^2_{n_x,1-\epsilon}}.
\)
By \eqref{eq:terminal_covariance_surrogate_bound},
\(
P_{t_f}\preceq \alpha \Sigma_{t_f}.
\)
Applying the congruence transformation with \(\Sigma_{t_f}^{-1/2}\) gives
\(
\Sigma_{t_f}^{-1/2}P_{t_f}\Sigma_{t_f}^{-1/2}\preceq \alpha I_{n_x}.
\)
Moreover,
\(P_{t_f}^{1/2}\Sigma_{t_f}^{-1}P_{t_f}^{1/2}\) and
\(\Sigma_{t_f}^{-1/2}P_{t_f}\Sigma_{t_f}^{-1/2}\) have the same nonzero eigenvalues; since both matrices are symmetric positive semidefinite, it follows that
\(
P_{t_f}^{1/2}\Sigma_{t_f}^{-1}P_{t_f}^{1/2}\preceq \alpha I_{n_x}.
\)
Substituting this bound into \eqref{eq:terminal_quadratic_form_mahalanobis} yields
\(
(x(t_f)-\mu_{t_f})^\top \Sigma_{t_f}^{-1}(x(t_f)-\mu_{t_f})
\le
\alpha\, z^\top z.
\)
Therefore,
$\mathbb P\!\left(x(t_f)\notin\mathcal X_{t_f}\right)
=
\mathbb P\!\left(
(x(t_f)-\mu_{t_f})^\top \Sigma_{t_f}^{-1}(x(t_f)-\mu_{t_f})>r_{t_f}^2
\right) \le
\mathbb P\!\left(\alpha\, z^\top z>r_{t_f}^2\right).$
Since \(z^\top z\sim \chi^2_{n_x}\),
\(
\mathbb P\!\left(\alpha\, z^\top z>r_{t_f}^2\right)
=
\mathbb P\!\left(
\chi^2_{n_x}>\frac{r_{t_f}^2}{\alpha}
\right)
=
\mathbb P\!\left(
\chi^2_{n_x}>\chi^2_{n_x,1-\epsilon}
\right)
=
\epsilon.
\)
Thus
\(
\mathbb P\!\left(x(t_f)\notin\mathcal X_{t_f}\right)\le \epsilon.
\)
\end{proof}

\textbf{Proof of Theorem~\ref{thm:cs_p2l_guarantee}}
\begin{proof}\label{proof:cs_p2l_guarantee}
Algorithm~\ref{alg:cs_p2l} is an instance of the P2L meta-algorithm
(Alg.~1 of~\cite{paccagnan2025pick}) with the following
identifications: the dataset $D$ corresponds to the realization
set~$Z_N$; the property $\varphi$ and violation measure $\ell$ define the
termination condition and selection rule. It remains to verify that the
policy construction defines a valid synthesis algorithm in the sense
of~\cite{paccagnan2025pick}, i.e., a map from a subset of $Z_N$ to a
policy. At iteration $r+1$, the policy is determined by
$\pi^{(r+1)} = \mathrm{Solve\text{-}CSP}(L(\theta^{(r)}, T^{(r+1)}))$.
Since $\theta^{(0)}$ is fixed and each subsequent $\theta^{(r)}$ is
obtained by applying $L$ along the nested sequence
$T^{(1)} \subset \cdots \subset T^{(r+1)}$, the entire
map from $T^{(r+1)}$ to $\pi^{(r+1)}$ is determined by $T^{(r+1)}$
and~$\theta^{(0)}$. By construction, the algorithm terminates when
$Y_{\pi^{(r)}}(z^{(i)}) = 1$ for all
$z^{(i)} \in Z_N \setminus T^{(r)}$, so that the compression property is
satisfied. The result then follows directly from Theorem~1
of~\cite{paccagnan2025pick}.
\end{proof}

\subsection{Spacecraft Powered-descent Dynamics}\label{appendix:dynamics}
We consider the drift in Simulator~\eqref{def:simulator} as:
$f(x_t,u_t,t; \lambda) = f_0(x_t,t; \lambda) + F(x_t)u_t,$
where $f_0:\mathbb{R}^{n_x}\times\mathbb{R}\mapsto\mathbb{R}^{n_x}$ is the natural dynamics and $F:\mathbb{R}^{n_x}\mapsto\mathbb{R}^{n_x \times 2}$ is a generic influence matrix that maps accelerations to changes in the state space. 
The natural dynamics are described as: $f_0(x_t,t) = f_{\text{kep}}(x_t) + F(x_t) a_{\text{pert}}(x,t),$
where $f_{\text{kep}}:\mathbb{R}^{n_x}\mapsto\mathbb{R}^{n_x}$ is Keplerian motion and $a_{\text{pert}}:\mathbb{R}^{n_x}\times\mathbb{R}\mapsto\mathbb{R}^{2}$ is a perturbing acceleration due to the rotating reference frame, solar radiation pressure,  higher-order gravitational effects, etc. Assuming Keplerian motion in a  non-rotating frame, these functions are
$f_{\text{kep}}(x_t) = \begin{bmatrix}
        v_t \\ 
        -\lambda(\mu / \|r_t\|_2^3) r_t
    \end{bmatrix},F = \begin{bmatrix}
         0_{2\times2}\\
         I_2
     \end{bmatrix},a_{\text{pert}}(x,t) = 0_{2\times 1}$
where $\mu\in\mathbb{R}$ is the gravitational parameter of the central body, $\lambda\in\mathbb{R}$ is the scaling of the gravitational parameter, and $r_t = \|r_t\|_2$.

\subsection{Alg.~\ref{alg:solve_csp} Setup}\label{appendix:CSP_setup}
The reference $\xi^{(0)} = \{x^{\mathrm{ref}}, \bar{u}^{\mathrm{ref}}, \zeta^{\mathrm{ref}}\}$
for Alg.~1 is initialized by the user and updated by the solution from Def.~\ref{def:surrogate_cs_problem}.
The initial mean state reference linearly interpolates between the initial and terminal means: \(
    x^{\mathrm{ref}}_k = (1-\tau_k)\,\mu_{t_0} + \tau_k\,\mu_{t_f},\)
where $\tau_k = k / K$. The feedforward control reference is set to zero, $\bar{u}^{\mathrm{ref}}_k = {0}$. The control covariance slack reference $\zeta^{\mathrm{ref}}_k$ is initialized as the midpoint of the feasible interval
$[0,\,\zeta^{\max}_k]$ such that $\zeta^{\mathrm{ref}}_k = \tfrac{1}{2}\,\zeta^{\max}_k$ and $\zeta^{\max}_k$ is the largest value consistent with the nominal control constraint in~\eqref{eq:cvx_control_constraint}, i.e.,
\(
    \zeta^{\max}_k
    = (u_{\max} - \|\bar{u}^{\mathrm{ref}}_k\|_2) /\sqrt{\chi^2_{n_u,\,1-\varepsilon_u}}.
\)
The fixed hyper-parameters are set as: $w_u=1.0,w_{\mathrm{vc}}=100,w_{\mathrm{tr}}=0.1$ and $\delta_{\mathrm{vc}}=\delta_{\mathrm{tr}}=10^{-6}$.






\bibliographystyle{IEEEtran} 
\bibliography{IEEEtranBST/IEEEabrv,IEEEtranBST/IEEEexample}


\end{document}